\documentclass[letterpaper, 10pt, conference]{ieeeconf}
\usepackage{graphicx}%
\usepackage{amsmath}%
\usepackage{amsfonts}%
\usepackage{amssymb}%

\IEEEoverridecommandlockouts
\def\boxit#1{\vbox{\hrule\hbox{\vrule\kern3pt
        \vbox{\kern3pt#1\kern3pt}\kern3pt\vrule}\hrule}}

\def\reals{ { {\rm  I \kern-0.15em R }  } }
\def\complex{ {\,{{\rm C} \kern-0.50em \raise0.20ex {  |}}\, }}

\def\Rbf{{\bf R}}

\def\Rxx{\Rbf_{\ssstyle X\kern-.1em X}}

\let\ssstyle=\scriptscriptstyle

\def\Kout{\setbox1=\hbox{\Huge\bf K}\hbox to
1.05\wd1{\hspace{.05\wd1}
}}
\def\Sout{\setbox1=\hbox{\Huge\bf S}\hbox to 1.05\wd1{\hspace{.05\wd1}
}}

\renewcommand{\QED}{\QEDopen}

\def\nn{{\nonumber}}

\newtheorem{lemma}{Lemma}
\newtheorem{theorem}{Theorem}

\renewcommand{\theenumi}{\roman{enumi}}

\title{\LARGE \bf Multi-channel Opportunistic Access: A Case of Restless Bandits with Multiple Plays}
\author{Sahand Haji Ali Ahmad, Mingyan Liu
\thanks{This work is partially supported by CNS-0238035 and CCF-0910765.  S. H. A. Ahmad and M. Liu are with the Dept. of Electrical Engineering and Computer Science, University of Michigan, Ann Arbor, MI 48105, \{shajiali, mingyan\}@eecs.umich.edu.}
}

\begin{document}

\maketitle
\thispagestyle{empty}
\pagestyle{empty}

\begin{abstract}
This paper considers the following stochastic control problem that arises in opportunistic spectrum access: a system consists of $n$ channels where the state (``good'' or ``bad'') of each channel evolves as independent and identically distributed Markov processes. A user can select exactly $k$ channels to sense and access (based on the sensing result) in each time slot. A reward is obtained whenever the user senses and accesses a ``good'' channel. The objective is to design a channel selection policy that maximizes the expected discounted total reward accrued over a finite or infinite horizon.  In our previous work we established the optimality of a greedy policy for the special case of $k=1$ (i.e., single channel access) under the condition that the channel state transitions are positively correlated over time.  In this paper we show under the same condition the greedy policy is optimal for the general case of $k\geq 1$; the methodology introduced here is thus more general.  This problem may be viewed as a special case of the restless bandit problem, with {\em multiple plays}. 
 We discuss connections between the current problem and existing literature on this class of problems.

\end{abstract}

\bibliographystyle{ieeetr}

\section{Introduction}\label{sec:intro}

We consider the following stochastic control problem: there are $n$ uncontrolled Markov chains, each an independent, identically-distributed, two-state discrete-time Markov process.  The two states will be denoted as state $1$ and state $0$ and the transition probabilities are given by $p_{ij}$, $i, j = 0, 1$.  

The system evolves in discrete time.  In each time instance, a user selects exactly $k$ out of the $n$ processes and is allowed to observe their states.  For each selected process that happens to be in state $1$ the user gets a reward; there is no penalty for selecting a channel that turns out to be state $0$ but each such occurrence represents a lost opportunity because the user is limited to selecting only $k$ of them.  The ones that the user does not select do not reveal their true states. Out objective is to derive a selection strategy whose total expected discounted rewarded over a finite or infinite horizon is maximized.

This is a Markov decision process (or MDP) problem \cite{puterman}.  Furthermore, it is a partially observed MDP (or POMDP) problem \cite{Smallwood&Sondik:71OR} due to the fact that the states of the underlying Markov processes are not fully observed at all times and that as a consequence the system {\em state} as perceived by the user is in the form of a probability distribution, also commonly referred to as the {\em information state} of the system \cite{kumar}.
This problem is also an instance of the restless bandit problem with multiple plays \cite{whittle,anantharam87,agrawal90}.  More discussion on this literature is provided in section \ref{sec:discussion}.

The application of the above problem abstraction to multichannel opportunistic access is as follows.  Each Markov process represents a wireless channel, whose state transitions reflect dynamic changes in channel conditions caused by fading, interference, and so on.  Specifically, we will consider state $1$ as the ``good'' state, in which a user (or transmitter) can successfully communicate with a receiver; state $0$ is the ``bad'' state, in which communication will fail.  The channel state is assumed to remain constant within a single discrete time step.  A multichannel system consists of $n$ distinct channels.  A user who wishes to use a particular channel at the beginning of a time step must first sense or probe the state of the channel, and can only transmit in a channel probed to be in the ``good'' state in the same time step.  The user cannot sense and access more than $k$ channels at a time due to hardware limitations.  If all $k$ selected channels turn out to be in the ``bad'' state, the user has to wait till the beginning of the next time step to repeat the selection process.

This model captures some of the essential features of multichannel opportunistic access as outlined above.  On the other hand, it has the following limitations: the simplicity of the iid two-state channel model; the implicit assumption that channel sensing is perfect and the lack of penalty if the user transmits in a bad channel due to imperfect sensing; and the assumption that the user can select an arbitrary set of $k$ channels out of $n$ (e.g., it may only be able to access a contiguous block of channels due to physical layer limitations).  Nevertheless this model does allow us to obtain analytical insights into the problem, and more importantly, some insight into the more general problem of restless bandits with multiple plays.

This model has been used and studied quite extensively in the past few years, mostly within the context of opportunistic spectrum access and cognitive radio networks, see for example \cite{Zhao&etal:08TWC,sahand-IT08,guha,guha2}.  \cite{Zhao&etal:08TWC} studied the same problem and proved the optimality of the greedy policy in the special case of $k=1, n=2$, \cite{Liu&Zhao:08Asilomar} proved the optimality of the greedy policy in the case of $k=n-1$, while \cite{guha,guha2} looked for provably good approximation algorithms for a similar problem.  Furthermore, the indexability (in the context of Whittle's heuristic index and indexability definition \cite{whittle}) of the underlying problem was studied in \cite{dahleh08,Liu&Zhao:08SDR}.

Our previous work \cite{sahand-IT08} established the optimality of the greedy policy for the special case of $k=1$ for arbitrary $n$ and under the condition $p_{11}\geq p_{01}$, i.e., when a channel's state transitions are positively correlated.  In this sense, the results reported in the present paper is a direct generalization of results in \cite{sahand-IT08}, as we shall prove the optimality of the greedy policy under the same condition but for any $n\geq k\geq 1$.  The main thought process used to prove this more general result derives from that used in \cite{sahand-IT08}.  However, there were considerable technical difficulties we had to overcome to reach the conclusion.

In the remainder of this paper we first formulate the problem in Section~\ref{sec:problem}, present preliminaries in Section \ref{sec:preliminaries}, and then prove the optimality of the greedy policy in Section~\ref{sec:optimal}.  We discuss our work within the context of restless bandit problems in Section~\ref{sec:discussion}.  Section~\ref{sec:conclusion} concludes the paper.

\section{Problem Formulation}\label{sec:problem}

As outlined in the introduction, we consider a user trying to access the wireless spectrum pre-divided into $n$ independent and statistically identical channels, each given by a two-state Markov chain.  The collection of $n$ channels is denoted by ${\cal N}$, each indexed by $i=1, 2, \cdots, n$.

The system operates in discrete time steps indexed by $t$, $t=1, 2, \cdots, T$, where $T$ is the time horizon of interest.  At time $t^{-}$, the channels go through state transitions, and at time $t$ the user makes the channel selection decision.  Specifically, at time $t$ the user selects $k$ of the $n$ channels to sense, the set denoted by $a^k \subset {\cal N}$.

For channels sensed to be in the ``good'' state (state $1$), the user transmits in those channels and collects one unit of reward for each such channel.  If none is sensed good, the user does not transmit, collects no reward, and waits until $t+1$ to make another choice. This process repeats sequentially until the time horizon expires.

The underlying system (i.e., the $n$ channels) is not fully observable to the user.  Specifically, channels go through state transition at time $t^{-}$ (or anytime between $(t-1, t)$), thus when the user makes the channel sensing decision at time $t$, it does not have the true state of any channel at time $t$.

Furthermore, upon its action (at time $t^{+}$) only $k$ channels reveal their true states.

The user's action space at time $t$ is given by the finite set $a^k(t) \subset {\cal N}$, where $a^k(t) = \{i_1,\ldots,i_K\}$.

We know (see e.g., \cite{Smallwood&Sondik:71OR,marcus,kumar}) that a sufficient statistic of such a system for optimal decision making, or the {\em information state} of the system \cite{marcus,kumar}, is given by the conditional probabilities of the state each channel is in given all past actions and observations.  Since each channel can be in one of two states, we denote this information state by ${\bar\omega}(t)=[\omega_1(t),\cdots,\omega_n(t)] \in [0, 1]^n$, where $\omega_i(t)$ is the conditional probability that channel $i$ is in state $1$ at time $t$ given all past states, actions and observations \footnote{Note that it is a standard way of turning a POMDP problem into a classic MDP problem by means of the information state, the main implication being that the state space is now uncountable.}.
Throughout the paper $\omega_i(t)$ will be referred to as the information state of channel $i$ at time $t$, or simply the channel probability of $i$ at time $t$.

Due to the Markovian nature of the channel model, the future information state is only a function of the current information state and the current action; i.e., it is independent of past history given the current information state and action.

It follows that the information state of the system evolves as follows.  Given that the state at time $t$ is ${\bar\omega}(t)$ and action $a^k(t)$ is taken, $\omega_i(t+1)$ for $i\in a^k(t)$ can take on two values: (1) $p_{11}$ if the observation is that channel $i$ is in a ``good'' state; 
this occurs with probability 
$\omega_i(t)$;
(2) $p_{01}$ if the observation is that channel $i$ is in a ``bad'' state;  
this occurs with probability 
$1-\omega_i$.
For any other channel $j\not\in a^k(t)$, with probability 1 the corresponding $\omega_j(t+1) = \tau(\omega_j(t))$ where the operator $\tau: [0,1] \rightarrow [0,1]$ is defined as
\begin{equation}
\tau(\omega) := \omega p_{11} + (1-\omega) p_{01},~~~ 0\le\omega\le 1 ~.
\end{equation}

The objective is to maximize its total discounted expected
reward over a finite horizon given in the following problem (P) (extension to infinite horizon is discussed in Section \ref{sec:discussion}):  

\begin{eqnarray}
\mbox{(P):}  ~~~ \max_{\pi} J_T^{\pi}({\bar\omega})
= \max_{\pi} E^{\pi} [ \sum_{t=1}^{T} \beta^{t-1} R_{\pi_t}({\bar\omega}(t)) | {\bar\omega}(1) = {\bar\omega}]\nn
\end{eqnarray}
where $0\leq \beta \leq 1$ is the discount factor, and $R_{\pi_t}({\bar\omega}(t))$ is the reward collected under state ${\bar\omega}(t)$ when channels in the set $a^k(t)=\pi_t(\bar\omega(t))$ are selected.  
The maximization in (P) is over the class of deterministic Markov policies \footnote{A Markov policy is a policy that derives its action only depending on the current (information) state, rather than the entire history of states, see e.g., \cite{kumar}.}.
An admissible policy $\pi$, given by the vector $\pi = [\pi_1, \pi_2, \cdots, \pi_T]$, is such that $\pi_t$ specifies a mapping from the current information state $\bar\omega(t)$ to a channel selection action $a^k(t)=\pi_t(\bar\omega(t)) \subset \{1, 2, \cdots, n\}$.
This is done without loss of optimality due to the Markovian nature of the underlying system, and due to known results on POMDPs \cite[Chapter 6]{kumar}.

\section{Preliminaries}\label{sec:preliminaries}

The dynamic programming (DP) representation of problem (P) is given as follows:
\begin{eqnarray}
V_{T}({\bar\omega}) &=& \max_{a^k \in {\cal N}, |a^k|=k} E[R_{a^k}({\bar\omega}) ] \nonumber\\
V_{t}({\bar\omega}) &=& \max_{a^k \in {\cal N}, |a^k|=k}

( \sum_{i \in a^k} \omega_i + 
 \beta \cdot \nonumber \\
 && \sum_{l_i \in \{0,1\}, i\in a^k} \left( \prod_{i\in a^k} \omega_{i}^{l_i}(1-\omega_{i})^{1-l_i} \right) \cdot \nn \\
&& V_{t+1}(p_{01}, \ldots, p_{01}, \tau(\omega_j), p_{11}, \ldots, p_{11})), \\
&& ~~~~~ t=1, 2, \cdots, T-1.  \label{DP-finite-t} \nonumber
\end{eqnarray}
In the last term, the channel state probability vector consists of three parts: a sequence of $p_{01}$'s that represent those channels sensed to be in state $0$ at time $t$ and the length of this sequence is the number of $l_i$'s equaling zero; a sequence of values $\tau(\omega_j)$ for all $j\not\in a^k$; and a sequence of $p_{11}$'s that represent those channels sensed to be in state $1$ at time $t$ and the length of this sequence is the number of $l_i$'s equaling one.  Note that the future expected reward is calculated by summing over all possible realizations of the $k$ selected channels.

The value function $V_t(\bar\omega)$ represents the maximum expected future reward that can be accrued starting from time $t$ when the information state is ${\bar\omega}$.  In particular, we have $V_1(\bar\omega) = \max_{\pi} J^{\pi}_T(\bar\omega)$, and an optimal deterministic Markov policy exists such that $a=\pi^*_t (\bar\omega)$ achieves the maximum in (\ref{DP-finite-t}) (see e.g., \cite{puterman} (Chapter 4)).  

For simplicity of representation, we introduce the following notations:
\begin{itemize}
\item $p_{01}[x]$: this is the vector $[p_{01}, p_{01}, \cdots, p_{01}]$ of length $x$;
\item $p_{11}[x]$: this is the vector $[p_{11}, p_{11}, \cdots, p_{11}]$ of length $x$.
\item We will use the notation:
\begin{eqnarray*}
 q(l_1, \cdots, l_{k}) := \prod_{1\leq i \leq k} \left( \omega_i^{l_i} (1-\omega_i)^{1-l_i}\right)
\end{eqnarray*}
for $l_1, \cdots, l_k \in \{0, 1\}$.  That is, given a vector of $0$s and $1$s (total of $k$ elements), $q()$ is the probability that a set of $k$ channels are in states given by the vector.
\end{itemize}

With the above notation, Eqn (\ref{DP-finite-t}) can be written as
\begin{eqnarray*}
V_{t}({\bar\omega}) &=& \max_{a^k \in {\cal N}, |a^k|=k}
( \sum_{i \in a^k} \omega_i +  \beta \cdot \\
&& \sum_{l_i \in \{0,1\}, i\in a^k} q(l_1, \cdots, l_k) \cdot\\

&& V_{t+1}(p_{01}[k- \sum{l_i}], \cdots, \tau(\omega_j), p_{11}[\sum{l_i}] )~.
\end{eqnarray*}

Solving (P) using the above recursive equation can be computationally heavy, especially considering the fact that $\bar\omega$ is a vector of probabilities.  It is thus common to consider suboptimal policies that are easier to compute and implement.  One of the simplest such heuristics is a greedy policy where at each time step we take an action that maximizes the immediate one-step reward. Our focus is to examine the optimality properties of such a simple greedy policy.

For problem (P), the greedy policy under state ${\bar\omega} = [\omega_1, \omega_2,
\cdots, \omega_n]$ is given by
\begin{equation}
a^k({\bar\omega})= 
\arg\max_{a^k \subset {\cal N}, |a^k|=k}  \sum_{i\in a^k} \omega_i ~.
\label{eq:a*}
\end{equation}
That is, the greedy policy seeks to maximize the reward {\em as if} there were only one step remaining in the horizon.  
In the next section we investigate the optimality of this policy.  Specifically, we will show that it is optimal in the case of $p_{11}\geq p_{01}$.  This extends the earlier result in \cite{sahand-IT08} that showed this to be true for the special case of $k=1$.

\section{Optimality of the Greedy Policy}
\label{sec:optimal}

In this section we show that the greedy policy is optimal when $p_{11} \geq p_{01}$. 
The main theorem of this section is as follows.
\begin{theorem} \label{mainThm}
The greedy policy is optimal for Problem (P) under the assumption that $p_{11}\geq p_{01}$.  That is, for $t=1, 2, \cdots, T$, $k\leq n$, and $\forall \bar\omega = [\omega_1, \cdots, \omega_n] \in [0, 1]^n$, we have
\begin{eqnarray}
V_t^k(\bar\omega; z^k(\bar\omega)) \geq V_t^k(\bar\omega; a^k), ~~~~~ \forall a^k \subset {\cal N}, 
\end{eqnarray}
where $z^k(\bar\omega)$ is the subset whose elements (indices) correspond to the $k$ largest values in $\bar\omega$, and $V_t^k(\bar\omega; a^k)$ the expected value of action $a^k$ followed by behaving optimally. 
\end{theorem}

Below we present a number of lemmas used in the proof of this theorem.  The first lemma introduces a notation that allows us to express the expected future reward under the greedy policy.

\begin{lemma}\label{lem:T}
There exist $T$ $n$-variable functions, denoted by $W_t^k(\bar\omega)$, $t=1, 2, \cdots, T$, each of which is a polynomial of order 1\footnote{Each function $W_t$ is affine in each variable, when all other variables are held constant.} and can be represented recursively in the following form:
\begin{eqnarray*}
&& W_T^k(\bar\omega) = \sum_{n-1+1\leq i\leq n} \omega_i \nonumber \\
&& W_t^k(\bar\omega) = \sum_{n-1+1\leq i\leq n} \omega_i + 
\beta \cdot \nonumber \\ 
&& \sum_{l_n, l_{n-1}, \cdots, l_{n+k-1} \in \{0, 1\}} 
q(l_n, \cdots, l_{n+k-1}) \cdot \nonumber\\
&& W_{t+1}^k(p_{01}[k-\sum l_i], \tau(\omega_i), \cdots, \tau(\omega_{n-k}), p_{11}[\sum l_i] ) ~. 
\end{eqnarray*}
\end{lemma}

\medskip
The proof is easily obtained using backward induction on $t$ given the recursive equation and noting
that 
the mapping $\tau()$ is linear.  The detailed proof is thus omitted for brevity. 

A few remarks are in order on this function $W_t^k(\bar\omega)$. 
\begin{enumerate}
\item Firstly, when $\bar\omega$ is given by an ordered vector $[\omega_1, \omega_2, \cdots, \omega_n]$ with $\omega_1\le\omega_2\le \cdots\le\omega_n$, $W_t^k(\bar\omega)$ is the expected total discounted future reward (from $t$ to $T$) by following the greedy policy.  
%

This follows from how the greedy policy works in the special case of $p_{11}\geq p_{01}$.  Note that in this case the conditional probability updating function $\tau(\omega)$ is a monotonically increasing function,  i.e., $\tau(\omega_1) \geq \tau(\omega_2)$ for $\omega_1 \geq \omega_2$.  Therefore the ordering of channel probabilities is preserved among those that are not observed.

If a channel has been observed to be in state ``1'' (respectively ``0''), its probability at the next step becomes $p_{11}\geq \tau(\omega)$ (respectively $p_{01}\leq \tau(\omega)$) for any $\omega\in[0, 1]$.  In other words, a channel observed to be in state ``1'' (respectively ``0'') will have the highest (respectively lowest) possible probability among all channels.

Therefore if we take the initial information state $\bar\omega(1)$, order the channels according to their probabilities $\omega_i(1)$, and sense the highest $k$ channels (top $k$ of the ordered list) with ties broken randomly, then following the greedy policy means that in subsequent steps we will keep a channel in its current position if it was sensed to be in state $1$ in the previous slot; otherwise, it was observed to be in state $0$ and gets thrown to the bottom of the ordered list.  The policy then selects the next top most (or rightmost) $k$ channels on this new ordered list.  This procedure is essentially the same as that given in the recursive expression of $W()$. 

\item Secondly, when $\bar\omega$ is not ordered, $W_t^k()$ reflects a policy that simply goes down the list of channels by the order fixed in $\bar\omega$, while each time tossing the ones observed to be $0$ to the end of the list and keeing those observed to be $1$ at the top of the list. 

\item Thirdly, the fact that $W^K_t$ is a polynomial of order 1 and affine in each of its elements implies that
\begin{eqnarray*}
&&W^K_t(\omega_1, \cdots, \omega_{n-2}, y, x) \\
&& - W^K_t(\omega_1,\cdots,\omega_{n-2}, x, y)\nn\\
& =& (x-y)[W^K_t(\omega_1,\cdots,\omega_{n-2}, 0, 1)- \\
&& ~~~~~ W^K_t(\omega_{1},\cdots,\omega_{n-2}, 1, 0)] ~.
\end{eqnarray*}
Similar results hold when we change the positions of $x$ and $y$. 
To see this, consider the above 
as two functions of $x$ and $y$, each having an $x$ term, a $y$ term, an $xy$ term and a constant term.  Since we are only swapping the positions of $x$ and $y$ in these two functions, the constant term remains the same,  and so does the $xy$ term. Thus the only difference is the $x$ term and the $y$ term, as given in the above equation.  This linearity result is used later in our proof.

\end{enumerate}

The next lemma establishes a sufficient condition for the optimality of the greedy policy.

\medskip
\begin{lemma}\label{lem:suff}
Consider Problem (P) under the assumption that $p_{11}\geq p_{01}$. To show that the greedy policy is optimal at time $t$ given that it is optimal at $t+1, t+2, \cdots, T$, it suffices to show that at time $t$ we have 
\begin{eqnarray}
&& W^k_t(\omega_1,\cdots,\omega_{j},x,y,\cdots,\omega_n)  \nn \\
&\leq&   W^k_t(\omega_1,\cdots,\omega_{j},y,x,\cdots,\omega_n), \label{eqn:suff}
\end{eqnarray}
for all $x\geq y$ and all $0\leq j \leq n-2$, with $j=0$ implying 
$W^k_t(x, y, \omega_3, \cdots, \omega_n)  \leq  W^k_t(y, x, \omega_3, \cdots, \omega_n)$. 
\end{lemma}

\begin{proof}
Since the greedy policy is optimal from $t+1$ on, it is sufficient to show that selecting the best $k$ channels followed by the greedy policy is better than selecting any other set of $k$ channels followed by the greedy policy.  If channels are ordered $\omega_1  \leq \cdots \leq  \omega_i \leq \cdots \leq \omega_n$ then the reward of the former is precisely given by $W^K_t(\omega_1,\hdots,\omega_n)$.  
On the other hand, the reward of selecting an arbitrary set $a^k$ of $k$ channels followed by acting greedily can be expressed as $W^k_t(\overline{a^k}, a^k)$, where $\overline{a^k}$ is the (increasingly) ordered set of channels not included in $a^k$.  
It remains to show that if Eqn (\ref{eqn:suff}) is true then we have 
$W^k_t(\overline{a^k}, a^k) \leq W^K_t(\omega_1,\hdots,\omega_n)$.  This is easily done since the ordered list ($\overline{a^k}, a^k$) may be converted to $\omega_1,\hdots,\omega_n$ through a sequence of switchings between two neighboring elements that are not increasingly ordered.  Each such switch invokes (\ref{eqn:suff}), thereby maintaining the ``$\leq$'' relationship. 
\end{proof}


\medskip
\begin{lemma}\label{lem:bound}
For $0 \leq \omega_1 \leq \omega_2 \leq \hdots \leq \omega_n \leq 1$, we have the following two inequalities for all $t=1, 2, \cdots, T$:
\begin{eqnarray*}
(A): && 1+W^k_t(\omega_2,\cdots,\omega_n,\omega_1)  \geq W^k_t(\omega_1,\cdots,\omega_n) \\
(B): && W_t^k(\omega_1, \cdots, \omega_j, y, x, \omega_{j+3}, \cdots, \omega_n) 
\geq \nn\\
&& W_t^k(\omega_1, \cdots, x, y, \omega_{j+3}, \cdots, \omega_n), 
\end{eqnarray*}
where $x\geq y$, $0\leq j\leq n-2$, and $j=0$ implies $W_t^k(y, x, \omega_3, \cdots, \omega_n) \geq W_t^k(x, y, \omega_3, \cdots, \omega_n)$. 
\end{lemma}

This lemma is the key to our main result and its proof, which uses a sample path argument, highly instructive.  It is however also lengthy, and for this reason has been relegated to the Appendix. 

\medskip
With the above lemmas, Theorem 1 is easily proven: 

{\em Proof of Theorem 1:}  We prove by induction on $T$.  When $t=T$, the greedy policy is obviously optimal. 
Suppose it is also optimal for all times $t+1, t+2, \cdots, T$, under the assumption $p_{11}\geq p_{01}$. 
Then at time $t$, by Lemma \ref{lem:suff}, it suffices to show that 
$W^k_t(\omega_1,\cdots,\omega_{j},x,y,\cdots,\omega_n)  \leq  W^k_t(\omega_1,\cdots,\omega_{j},y,x,\cdots,\omega_n)$ for all $x\geq y$ and $0\leq j \leq n-2$.  But this is proven in Lemma \ref{lem:bound}. 
\hspace*{\fill}~\QED\par



\section{Discussion} \label{sec:discussion}

While the formulation (P) is a finite horizon problem, the same result applies to the infinite horizon discounted reward case using standard techniques as we have done in our previous work \cite{icc-08,sahand-IT08}. 

In the case of infinite horizon, the problem studied in this paper is closely associated with the class of multi-armed bandit problems \cite{gittins} and restless bandit problems \cite{whittle}.  This is a class of problems where $n$ controlled Markov chains (also called machines or arms) are activated (or played) one at a time.  A machine when activated generates a state dependent reward and moves to the next state according to a Markov rule.  A machine not activated either stays frozen in its current state (a rested bandit) or moves to the next state according to a possibly different Markov rule (a restless bandit).  The problem is to decide the sequence in which these machines are activated so as to maximize the expected (discounted or average) reward over an infinite horizon.  

The multi-armed bandit problem was originally solved by Gittins (see \cite{gittins}), who showed that there exists an {\em index} associated with each machine that is solely a function of that individual machine and its state, and that playing the machine currently with the highest index is optimal.  This index has since been referred to as the {\em Gittins index}.  The remarkable nature of this result lies in the fact that it decomposes the $n$-dimensional problem into $n$ 1-dimensional problems, as an index is defined for a machine independent of others.  The restless bandit problem on the other hand was proven much more complex, and is PSPACE-hard in general \cite{tsitsiklis}.
Relatively little is known about the structure of its optimal policy in general.  In particular, the Gittins index policy is not in general optimal \cite{whittle}. 


When multiple machines are activated simultaneously, the resulting problem is referred to as multi-armed bandits with {\em multiple plays}.  Again optimal solutions to this class of problems are not known in general.  A natural extension to the Gittins index policy in this case is to play the machines with the highest Gittins indices (this will be referred to as the {\em extended Gittins index policy} below).  This is not in general optimal for multi-armed bandits with multiple plays and an infinite horizon discounted reward criterion, see e.g., \cite{ishikida92,pandelis99}.  However, it may be optimal in some cases, see e.g., \cite{pandelis99} for conditions on the reward function, and \cite{song02} for an undiscounted case where the Gittins index is always achieved at time 1.  Even less is known when the bandits are restless, though 
asymptotic results for restless bandits with multiple plays were provided in \cite{whittle} and \cite{weber}. 

The problem studied in the present paper is an instance of the restless bandits with multiple plays (in the infinite horizon case).   Therefore what we have shown in this paper is an instance of the restless bandits problem with multiple plays, for which the extended Gittins index policy is optimal.

\section{Conclusion} \label{sec:conclusion}

In this paper we studied a stochastic control problem that arose in opportunistic spectrum access.  A user can sense and access $k$ out of $n$ channels at a time and must select judiciously in order to maximize its reward.  We extend a previous result where a greedy policy was shown to be optimal in the special case of $k=1$ under the condition that the channel state transitions are positively correlated over time.  In this paper we showed that under the same condition the greedy policy is optimal for the general case of $k\geq 1$.   This result also contributes to the understanding of the class of restless bandit problems with multiple plays.

\bibliography{../osa-IT}

\begin{appendix}


{\em Proof of Lemma 3:} 
We would like to show
\begin{eqnarray*}
(A): && 1+W^k_t(\omega_2,\cdots,\omega_n,\omega_1)  \geq W^k_t(\omega_1,\cdots,\omega_n)  \\
(B): && W_t^k(\omega_1, \cdots, \omega_j, y, x, \omega_{j+3}, \cdots, \omega_n) 
\geq \\
&& W_t^k(\omega_1, \cdots, x, y, \omega_{j+3}, \cdots, \omega_n), 
\end{eqnarray*}
where $x\geq y$, $0\leq j\leq n-2$, and $j=0$ implies $W_t^k(y, x, \omega_3, \cdots, \omega_n) \geq W_t^k(x, y, \omega_3, \cdots, \omega_n)$. 

The two inequalities (A) and (B) will be shown together using an induction on $t$.  For $t=T$, part (A) is true because
$LHS = 1 + \omega_1 + \sum_{i=n-k+2}^{n} \omega_i \geq \omega_{n-k+1} + \sum_{i=n-k+2}^{n} \omega_i = RHS$. 
Part (B) is obviously true for $t=T$ since $x\geq y$. 

Suppose (A) and (B) are both true for $t+1, \cdots, T$.  Consider time $t$, and we will prove (A) first. 
Note that in the next step, channel 1 is selected by the action on the LHS of (A) but not by the RHS, while channel $n-k+1$ is selected by the RHS of (A) but not by the LHS.  Other than this difference both sides select the same set of channels indexed $n-k+2, \cdots, n$.  We now consider four possible cases in terms of the realizations of channels 1 and $n-k+1$. 

Case (A.1): channels 1 and $n-k+1$ have the state realizations ``0'' and ``1'', respectively.  

We will use a sample-path argument. Note that while these two channels are not both observed by either side, the realizations hold for the underlying sample path regardless.  In particular, even though the LHS does not select channel $n-k+1$ and therefore does not get to actually observe the realization of ``1'', the fact remains that channel $n-k+1$ is indeed in state 1 under this realization, and therefore its future expected reward must reflect this.  It follows that under this realization channel $n-k+1$ will have probability $p_{11}$ for the next time step even though we did not get to observe the state 1.  The same is true for the RHS.  This argument applies to the other three cases and is thus not repeated. 

Conditioned on this realization, the LHS and RHS are evaluated as follows (denoted as $\{LHS|_{(0,1)}\}$ and $\{RHS|_{(0,1)}\}$, respectively): 
\begin{eqnarray*}
&& \{LHS|_{(0,1)}\} \\
&=& 1 + \sum_{n-k+2\leq i\leq n} \omega_i + 
\beta \cdot \\
&& \sum_{l_{n-k+2}, \cdots, l_n \in\{0,1\}} 
q(l_{n-k+2}, \cdots, l_n) \cdot  \\
&&  W_{t+1}^k(p_{01}[k-\sum l_i], \tau(\omega_2), \cdots, \\
&&  ~~~~~~ \tau(\omega_{n-k+1})=p_{11}, p_{11}[\sum l_i] ) ~~; 
\end{eqnarray*}

\begin{eqnarray*}
&& \{RHS|_{(0,1)}\} \\
&=& 1 + \sum_{n-k+2\leq i\leq n} \omega_i + 
\beta \cdot \\
&& \sum_{l_{n-k+2}, \cdots, l_n \in\{0,1\}} 
q(l_{n-k+2}, \cdots, l_n) \cdot \\
&& W_{t+1}^k(p_{01}[k-\sum l_i -1], \tau(\omega_1)=p_{00},  \\
&& ~~~~~~ \tau(\omega_2), \cdots, \tau(\omega_{n-k}), p_{11}[\sum l_i +1] ) \nonumber \\
& =& \{LHS|_{(0,1)}\} 
\end{eqnarray*}

Case (A.2): channels 1 and $n-1+1$ have the state realizations ``1'' and ``1'', respectively.  
\begin{eqnarray*}
&& \{LHS|_{(1,1)}\} \\
&=& 1 + 1 + \sum_{n-k+2\leq i\leq n} \omega_i + 
\beta \cdot \\
&& \sum_{l_{n-k+2}, \cdots, l_n \in\{0,1\}} 
q(l_{n-k+2}, \cdots, l_n) \cdot \nonumber \\
&& W_{t+1}^k(p_{01}[k-\sum l_i -1], \tau(\omega_2), \cdots, \nn \\
&& ~~~~~~ \tau(\omega_{n-k+1})=p_{11}, p_{11}[\sum l_i +1] )~; 
\end{eqnarray*}

\begin{eqnarray*}
&& \{RHS|_{(1,1)}\} \\
&=& 1 + \sum_{n-k+2\leq i\leq n} \omega_i + 
\beta \cdot \nn \\
&& \sum_{l_{n-k+2}, \cdots, l_n \in\{0,1\}} 
q(l_{n-k+2}, \cdots, l_n) \cdot \nonumber \\
&& W_{t+1}^k(p_{01}[k-\sum l_i -1], \tau(\omega_1)=p_{11}, \nn \\
&& ~~~~~~ \tau(\omega_2), \cdots, \tau(\omega_{n-k}), p_{11}[\sum l_i +1] ) \nonumber \\
&\leq& 1 + \sum_{n-k+2\leq i\leq n} \omega_i + 
\beta \cdot\\
&& \sum_{l_{n-k+2}, \cdots, l_n \in\{0,1\}} 
q(l_{n-k+2}, \cdots, l_n) \cdot \nonumber \\
&& W_{t+1}^k(p_{01}[k-\sum l_i -1], \tau(\omega_2), \cdots, \tau(\omega_{n-k}), \\
&& ~~~~~~ p_{11}, p_{11}[\sum l_i +1] ) \nonumber \\
&=& \{LHS|_{(1,1)}\} -1 \leq \{LHS|_{(1,1)}\} 
\end{eqnarray*}
where the first inequality is due to the induction hypothesis of (B).

Case (A.3): channels 1 and $n-1+1$ have the state realizations ``0'' and ``0'', respectively.  
\begin{eqnarray*}
&& \{RHS|_{(0,0)}\} \\
&=& \sum_{n-k+2\leq i\leq n} \omega_i + 
\beta \cdot \\
&& \sum_{l_{n-k+2}, \cdots, l_n \in\{0,1\}} 
q(l_{n-k+2}, \cdots, l_n) \cdot \nonumber \\
&& W_{t+1}^k(p_{01}[k-\sum l_i ], \tau(\omega_1)=p_{01}, \tau(\omega_2), \cdots, \tau(\omega_{n-k}),\\
&& ~~~~~~ p_{11}[\sum l_i] ) ~; 
\end{eqnarray*}

\begin{eqnarray*}
&& \{LHS|_{(0,0)}\} \\
&=& 1 + \sum_{n-k+2\leq i\leq n} \omega_i + 
\beta \cdot \\
&& \sum_{l_{n-k+2}, \cdots, l_n \in\{0,1\}} 
q(l_{n-k+2}, \cdots, l_n) \cdot \nonumber \\
&& W_{t+1}^k(p_{01}[k-\sum l_i], \tau(\omega_2), \cdots, \tau(\omega_{n-k}), \\
&& ~~~~~~ \tau(\omega_{n-k+1})=p_{01}, p_{11}[\sum l_i] ) \nonumber \\
&\geq& 1 + \sum_{n-k+2\leq i\leq n} \omega_i + 
\beta \cdot \\
&& \sum_{l_{n-k+2}, \cdots, l_n \in\{0,1\}} 
q(l_{n-k+2}, \cdots, l_n) \cdot \nonumber \\
&& W_{t+1}^k(p_{01}[k-\sum l_i], \tau(\omega_2), \cdots, \tau(\omega_{n-k}), \\
&& ~~~~~~ p_{11}[\sum l_i], p_{01} ) \nonumber \\
&\geq& \sum_{n-k+2\leq i\leq n} \omega_i + 
\beta \cdot \\
&& \sum_{l_{n-k+2}, \cdots, l_n \in\{0,1\}} 
q(l_{n-k+2}, \cdots, l_n) \cdot \nonumber \\
&& \left( 1+ W_{t+1}^k(p_{01}[k-\sum l_i], \tau(\omega_2), \cdots, \tau(\omega_{n-k}), \right.\\
&& ~~~~~~ \left. p_{11}[\sum l_i], p_{01} ) \right) \nonumber \\
&\geq& \sum_{n-k+2\leq i\leq n} \omega_i + 
\beta \cdot \\
&& \sum_{l_{n-k+2}, \cdots, l_n \in\{0,1\}} 
q(l_{n-k+2}, \cdots, l_n) \cdot \nonumber \\
&& W_{t+1}^k(p_{01}, p_{01}[k-\sum l_i], \tau(\omega_2), \cdots, \tau(\omega_{n-k}), \\
&& p_{11}[\sum l_i]) \\
&=& \{RHS|_{(0,0)}\}  
\end{eqnarray*}
where the first inequality is due to the induction hypothesis of (B), the last inequality due to the induction hypothesis of (A).  Also, the second inequality utilizes the total probability over the distribution $q(l_{n-k+2}, \cdots, l_n)$ and the fact that $\beta\leq 1$. 

Case (A.4): channels 1 and $n-1+1$ have the state realizations ``1'' and ``0'', respectively.  
\begin{eqnarray*}
&& \{RHS|_{(1,0)}\} \\
&=& \sum_{n-k+2\leq i\leq n} \omega_i + 
\beta \cdot \\
&& \sum_{l_{n-k+2}, \cdots, l_n \in\{0,1\}} 
q(l_{n-k+2}, \cdots, l_n) \cdot \nonumber \\
&& W_{t+1}^k(p_{01}[k-\sum l_i], \tau(\omega_1)=p_{11}, \tau(\omega_2), \cdots, \\
&& ~~~~~~ \tau(\omega_{n-k}), p_{11}[\sum l_i] )
\end{eqnarray*}

\begin{eqnarray*}
&& \{LHS|_{(1,0)}\} \\
&=& 1 + 1 + \sum_{n-k+2\leq i\leq n} \omega_i + 
\beta \cdot \\
&& \sum_{l_{n-k+2}, \cdots, l_n \in\{0,1\}} 
q(l_{n-k+2}, \cdots, l_n) \cdot \nonumber \\
&& W_{t+1}^k(p_{01}[k-\sum l_i -1], \tau(\omega_2), \cdots, \tau(\omega_{n-k}), \\
&& ~~~~~~ \tau(\omega_{n-k+1})=p_{01}, p_{11}[\sum l_i+1] ) \nonumber \\
&\geq& 1 + 1 + \sum_{n-k+2\leq i\leq n} \omega_i + 
\beta \cdot \\
&& \sum_{l_{n-k+2}, \cdots, l_n \in\{0,1\}} 
q(l_{n-k+2}, \cdots, l_n) \cdot \nonumber \\
&& W_{t+1}^k(p_{01}[k-\sum l_i - 1], \tau(\omega_2), \cdots, \tau(\omega_{n-k}), \\
&& ~~~~~~ p_{11}[\sum l_i + 1], p_{01} ) \nonumber \\
&\geq& 1+ \sum_{n-k+2\leq i\leq n} \omega_i + 
\beta \cdot \\
&& \sum_{l_{n-k+2}, \cdots, l_n \in\{0,1\}} 
q(l_{n-k+2}, \cdots, l_n) \cdot \nonumber \\
&& \left( 1+ W_{t+1}^k(p_{01}[k-\sum l_i-1], \tau(\omega_2), \cdots, \right. \\
&& ~~~~~~ \left. \tau(\omega_{n-k}), p_{11}[\sum l_i + 1], p_{01} ) \right) \nonumber \\
&\geq& 1+ \sum_{n-k+2\leq i\leq n} \omega_i + 
\beta \cdot \\
&& \sum_{l_{n-k+2}, \cdots, l_n \in\{0,1\}} 
q(l_{n-k+2}, \cdots, l_n) \cdot \nonumber \\
&& W_{t+1}^k(p_{01}[k-\sum l_i], \tau(\omega_2), \cdots, \tau(\omega_{n-k}), \\
&& p_{11}[\sum l_i + 1] \nonumber \\
&\geq& 1+ \sum_{n-k+2\leq i\leq n} \omega_i + 
\beta \cdot \\
&& \sum_{l_{n-k+2}, \cdots, l_n \in\{0,1\}} 
q(l_{n-k+2}, \cdots, l_n) \cdot \nonumber \\
&& W_{t+1}^k(p_{01}[k-\sum l_i ], p_{11}, \tau(\omega_2), \cdots, \tau(\omega_{n-k}), \\
&& p_{11}[\sum l_i ] ) \nonumber \\
&=& 1+ \{RHS|_{(1,0)}\} \geq   \{RHS|_{(1,0)}\}
\end{eqnarray*}
where the first and last inequalities are due to the induction hypothesis of (B), the third due to the induction hypothesis of (A).  

With these four cases, we conclude the induction step of proving (A).  We next prove the induction step of (B). 
%
We consider three cases in terms of whether $x$ and $y$ are among the top $k$ channels to be selected in the next step. 

Case (B.1): both $x$ and $y$ belong to the top $k$ positions on both sides.  In this case there is no difference between the LHS and RHS along each sample path, since both channels will be selected and the result will be the same. 

Case (B.2): neither $x$ nor $y$ is among the top $k$ positions on either side.  This implies that $j\leq n-k-2$. 
We have: 
\begin{eqnarray*}
&& LHS \\
&=& \sum_{n-k+2\leq i\leq n} \omega_i + 
\beta \cdot\\
&& \sum_{l_{n-k+2}, \cdots, l_n \in\{0,1\}} 
q(l_{n-k+2}, \cdots, l_n) \cdot \nonumber \\
&& W_{t+1}^k(p_{01}[k-\sum l_i ], \tau(\omega_1), \cdots, \tau(\omega_j), \\
&& \tau(y), \tau(x), \tau(\omega_{j+3}), 
\cdots, p_{11}[\sum l_i] )~;
\end{eqnarray*}

\begin{eqnarray*}
&& RHS \\
&=& \sum_{n-k+2\leq i\leq n} \omega_i + 
\beta \cdot \\
&& \sum_{l_{n-k+2}, \cdots, l_n \in\{0,1\}} 
q(l_{n-k+2}, \cdots, l_n) \cdot \nonumber \\
&& W_{t+1}^k(p_{01}[k-\sum l_i ], \tau(\omega_1), \cdots, \tau(\omega_j), \\
&& \tau(x), \tau(y), \tau(\omega_{j+3}), 
\cdots, p_{11}[\sum l_i] ) \nonumber \\
&\geq& LHS
\end{eqnarray*}
where the last inequality is due to the monotonicity of $\tau()$ and the induction hypothesis of (B). 

Case (B.3): exactly one of the two belongs to the the top $k$ channels on each side.  This implies that $j=n-k-1$. 
By the linearity of the function $W^k_t$ we have the following: 
\begin{eqnarray}
&& W^k_t(\omega_1,\cdots, \omega_{n-k-1},y,x,\omega_{n-k+2},\cdots,\omega_n) \nn \\
&& - W^k_t(\omega_1,\cdots, \omega_{n-k-1},x, y ,\omega_{n-k+2},\cdots,\omega_n) \nn \\
&=& (x - y)  (W^k_t(\omega_1,\cdots, \omega_{n-k-1},0,1,\omega_{n-k+2},\cdots,\omega_n) - \nn \\
&& W^k_t(\omega_1,\cdots, \omega_{n-k-1},1,0,\omega_{n-k+2},\cdots,\omega_n) )\label{eqn:B3}
\end{eqnarray}
However, we have 
\begin{eqnarray*}
%
&& W^k_t(\omega_1,\cdots, \omega_{n-k-1},1,0,\omega_{n-k+2},\cdots,\omega_n) \\
&=& \sum_{n-k+2\leq i\leq n} \omega_i + 
\beta \cdot \\
&& \sum_{l_{n-k+2}, \cdots, l_n \in\{0,1\}} 
q(l_{n-k+2}, \cdots, l_n) \cdot \nonumber \\
 && W_{t+1}^k(p_{01}[k-\sum l_i], \tau(\omega_1), \cdots, \tau(\omega_{n-k-1}), \\
 && p_{11}, p_{11}[\sum l_i] ) \\
 &\leq& \sum_{n-k+2\leq i\leq n} \omega_i + 
\beta \cdot \\
&& \sum_{l_{n-k+2}, \cdots, l_n \in\{0,1\}} 
q(l_{n-k+2}, \cdots, l_n) \cdot \nonumber \\
 && \left( 1+ W_{t+1}^k(p_{01}[k-\sum l_i -1], \tau(\omega_1), \cdots,\right.\\
 && ~~~~~~ \left. \tau(\omega_{n-k-1}), p_{11}[\sum l_i + 1], p_{01} ) \right) \\
 &\leq& \sum_{n-k+2\leq i\leq n} \omega_i + 
\beta \cdot \\
&& \sum_{l_{n-k+2}, \cdots, l_n \in\{0,1\}} 
q(l_{n-k+2}, \cdots, l_n) \cdot \nonumber \\
 && \left( 1+ W_{t+1}^k(p_{01}[k-\sum l_i -1], \tau(\omega_1), \cdots,\right. \\
 && ~~~~~~ \left. \tau(\omega_{n-k-1}), p_{01}, p_{11}[\sum l_i+1]) \right) \\
 &\leq& 1+ \sum_{n-k+2\leq i\leq n} \omega_i + 
\beta \cdot \\
&& \sum_{l_{n-k+2}, \cdots, l_n \in\{0,1\}} 
q(l_{n-k+2}, \cdots, l_n) \cdot \nonumber \\
&& W_{t+1}^k(p_{01}[k-\sum l_i -1], \tau(\omega_1), \cdots, \tau(\omega_{n-k-1}), \\
&& ~~~~~~ p_{01}, p_{11}[\sum l_i + 1] ) \nonumber \\
&=& W^k_t(\omega_1,\cdots, \omega_{n-k-1},0,1,\omega_{n-k+2},\cdots,\omega_n)
\end{eqnarray*}
Since $x\geq y$, we have $LHS\geq RHS$ in Eqn (\ref{eqn:B3}).  This concludes the induction step of (B). 


\end{appendix}

\end{document}